\documentclass{article}
\usepackage{times}
\usepackage{helvet}
\usepackage{courier}
\usepackage[colorlinks=true]{hyperref}
\usepackage{graphicx}

\usepackage{nicefrac}
\usepackage{dsfont}
\usepackage{amsmath,amssymb,amsfonts,amsthm}
\allowdisplaybreaks

\usepackage{authblk}

\usepackage{xspace}
\newcommand{\TwoChoices}{\mbox{\textsc{2-Choices}\xspace}}
\newcommand{\voter}{\mbox{\textsc{Voter}\xspace}}
\newcommand{\MaxLPA}{\mbox{\textsc{Max-LPA}\xspace}}
\newcommand{\LPA}{\mbox{\textsc{LPA}\xspace}}

\newcommand{\Ex}[1]{\mathbf{E}\left[ #1 \right]}
\newcommand{\cond}{\ \middle| \ }

\newcommand{\bone}{\mathds{1} }
\newcommand{\bx}{\mathbf{x} }
\newcommand{\bc}[1]{\mathbf{c}^{(#1)}}
\newcommand{\bb}{\mathbf{\bone_{B}} }
\newcommand{\bv}{\mathbf{v} }
\newcommand{\bw}{\mathbf{w} }
\newcommand{\Prob}[2]{\mathbf{P}_{#1} \left( #2 \right)}
\newcommand{\norm}[1]{\Vert #1 \Vert}
\newcommand{\bigO}{\mathcal{O}}

\newtheorem{definition}{Definition}
\newtheorem{theorem}{Theorem}
\newtheorem{lemma}{Lemma}
\newtheorem{corollary}{Corollary}
\newtheorem{claim}{Claim}

\usepackage{thm-restate}

\begin{document}
\title{Distributed Community Detection via\\ Metastability of the 2-Choices Dynamics}

\author{Emilio Cruciani}
\affil{%
 Gran Sasso Science Institute, L'Aquila, Italy\authorcr
 \texttt{emilio.cruciani@gssi.it}
}

\author{Emanuele Natale}
\affil{%
 Max Planck Institute for Informatics, Saarbr\"ucken, Germany
 \& Universit\'e C\^ote d'Azur, CNRS, I3S, Inria\authorcr
 \texttt{enatale@mpi-inf.mpg.de}
}

\author{Giacomo Scornavacca}
\affil{%
 University of L'Aquila, L'Aquila, Italy\authorcr
 \texttt{giacomo.scornavacca@graduate.univaq.it}
}

\date{}

\maketitle

\begin{abstract}
We investigate the behavior of a simple majority dynamics on networks of agents whose interaction topology exhibits a community structure. 
By leveraging recent advancements in the analysis of dynamics, we prove that, when the states of the nodes are randomly initialized, the system rapidly and stably converges to a configuration in which the communities maintain internal consensus on different states. 
This is the first analytical result on the behavior of dynamics for non-consensus problems on non-complete topologies, based on the first symmetry-breaking analysis in such setting. 

Our result has several implications in different contexts in which dynamics are adopted for computational and biological modeling purposes. 
In the context of \emph{Label Propagation Algorithms}, a class of widely used heuristics for \emph{community detection}, it represents the first theoretical result on the behavior of a distributed label propagation algorithm with quasi-linear message complexity.
In the context of \emph{evolutionary biology}, dynamics such as the Moran process have been used to model the spread of mutations in genetic populations~\cite{lieberman_evolutionary_2005};
our result shows that, when the probability of adoption of a given mutation by a node of the evolutionary graph depends super-linearly on the frequency of the mutation in the neighborhood of the node and the underlying evolutionary graph exhibits a community structure, there is a non-negligible probability for \emph{species differentiation} to occur.
\end{abstract}
\clearpage

\section{Introduction}

Dynamics are simple stochastic processes on networks, in which agents update their own state according to a symmetric function of the state of their neighbors and of their current state, with no dependency on time or on the topology of the network~\cite{mossel_opinion_2017,emanuele_natale_computational_2017}. 
In previous decades, in the context of automata networks, this kind of systems has been investigated from a computability point of view, attracting the interest of mathematicians and physicists. 
Recently it has been subject to a renewed interest from computer scientists, as new techniques for analyzing this class of processes have made possible to answer questions regarding their efficiency and capability as distributed algorithms~\cite{Doerr:2011,becchetti_plurality_2015,becchetti_stabilizing_2016,becchetti_simple_2017,cooper_fast_2015,cooper_fast_2016}. 

In this work we consider the \TwoChoices{} dynamics (Definition~\ref{def:2choices}), in which at each discrete-time step each agent samples two random neighbors with replacement and, if the two have the same state, the agent adopts that state. 
The process rapidly converges to \emph{consensus}, i.e., a configuration where all agents have the same state, if the proportion of agents supporting one state exceeds a given function of the second eigenvalue of the graph~\cite{cooper_fast_2015,cooper_fast_2016}.
Their proofs leverage an interesting property of the \TwoChoices{} dynamics, i.e., that the expected number of agents supporting one state can be expressed as a quadratic form of the transition matrix of a simple random walk on the underlying graph. 
This fact allows to relate the behavior of the process to the eigenspaces of the graph. 

Motivated by questions arising in \emph{graph clustering} and \emph{evolutionary biology}, we exploit the aforementioned relation to show a more fine-grained understanding of the \emph{consensus} behavior of the \TwoChoices{} dynamics. 
Our new analysis combines symmetry-breaking techniques~\cite{becchetti_stabilizing_2016,clementi_tight_2017} and concentration of probability arguments with a linear algebraic approach~\cite{cooper_fast_2015,cooper_fast_2016} to obtain the first symmetry-breaking analysis for dynamics on non-complete topologies.

\vspace{1em}
\fbox{\parbox{0.9\columnwidth}{
\textbf{Informal description of Theorem~\ref{theorem:main}.}
Let the agents of a network initially pick a random binary state and then run the \TwoChoices{} dynamics.
If the network has a \emph{community structure} there is a significant probability that it will rapidly converge to an \emph{almost-clustered} configuration, where almost all nodes within each community share the same state, but the predominant states in the communities are different. 
In other words, with constant probability, after a short time the states of the nodes constitute a labeling which reveals the clustered structure of the network.}}
\vspace{1em}

The aforementioned probability for the labeling to reveal the community structure can be amplified via \emph{Community-Sensitive Labeling}~\cite{becchetti_simple_2017},
transforming the \TwoChoices{} dynamics into a \emph{distributed label propagation algorithm} with quasi-linear message complexity.

We remark that, because of the stochastic and time-independent behavior of the \TwoChoices{} dynamics, the process eventually leaves almost-clustered configurations and reaches a \emph{monochromatic} configuration in which all agents have the same state. 
However, before that happens, we prove that the process remains in almost-clustered configurations for a time equal to a large-degree polynomial in $n$. 
Hence, the event that the process leaves the almost-clustered configuration is negligible for most practical applications. 
This key transitory property of some stochastic processes, 
called \emph{metastability}~\cite{auletta_metastability_2012,ferraioli_metastability_2015}, 
has recently attracted a lot of attention in the Theoretical Computer Science community.

\subsection{Label Propagation Algorithms}

\emph{Label Propagation Algorithms} (LPAs) are a widely used class of algorithms used for \emph{community detection} and inspired by epidemic processes on networks.
The generic pattern of such algorithms can be described as follows:
First, a label taken from a finite set is assigned to each node according to some \emph{initialization rule};
then the nodes are activated following some \emph{activation rule};
active nodes interact with their neighbors and update their labels according to some \emph{local majority-based update rule}.

After the first algorithm, known in literature as \LPA{}, has been proposed and its effectiveness empirically assessed~\cite{raghavan2007near}, a new line of research started with the goal of improving the quality of the detected communities and the efficiency of the algorithm~\cite{leung_towards_2009,liu_advanced_2010,boldi_layered_2011,vsubelj2011robust,vsubelj2011unfolding,xie_labelrank_2013,zhang2017label}, and to investigate more general settings, e.g., dynamic networks~\cite{xie2013labelrankt,clementi_distributed_2013}.
Many variants with small variations on initialization rule, activation rule, and local update rule have been proposed, but they have only been validated experimentally.
On the other hand, there exist only few theoretical works. 
One shows the equivalence of \LPA{} with finding the minima of a generalization of the Ising model, used in statistical mechanics to describe the spin interaction of electrons on a crystalline lattice~\cite{tibely_equivalence_2008}.
Another is the first and only rigorous analysis of a variant of \LPA{} on the \emph{Stochastic Block Model}\footnote{The \emph{Stochastic Block Model} is a generative model for random graphs, that produces graphs with community structure.}~\cite{kothapalli2013analysis}:
They propose \MaxLPA{}, i.e., a synchronous version of LPA that follows a deterministic majority rule, and analyze its behavior on $\mathcal{G}_{2n,p,q}$ graphs with parameters $p=\Omega(n^{-\nicefrac{1}{4}+\varepsilon})$ and $q=\bigO(p^2)$, i.e., on graphs that present very dense communities of constant diameter separated by a sparse cut.

The absence of substantial theoretical progress in the analysis of LPAs is largely due to the lack of techniques for handling the interplay between the non-linearity of the local update rules and the topology of the graph. 
In this work we look at the \TwoChoices{} dynamics as a \emph{distributed label propagation algorithm}. 
The randomized nature of the \TwoChoices{} dynamics introduces a major challenge with respect to deterministic rules such as the one of \MaxLPA{}.

\subsubsection{Comparison with our result.}
Let $a$ and $b$ respectively be the number of neighbors of each agent in its own community and in the other community; let $d:=a+b$.
The analysis of \MaxLPA{}~\cite{kothapalli2013analysis} essentially requires $a \geq n^{\nicefrac{3}{4}-\varepsilon}$ and $b\leq c a^2 / n$, for some arbitrary constants $\varepsilon$ and $c$.
Our analysis requires\footnote{$\lambda$ is the maximum eigenvalue, in absolute value and different from $1$, of the transition matrices of the subgraphs induced by the communities.} $\lambda \leq n^{-\nicefrac{1}{4}}$, which implies $a \geq n^{\nicefrac{1}{2}}$ because of the extremality of Ramanujan graphs, and $b/d \leq n^{-\nicefrac{1}{2}}$.
Compared to the analysis of \MaxLPA{}, Theorem~\ref{theorem:main} holds for much sparser communities at the price of a stricter condition on the cut.
Moreover, given the distributed nature of the two algorithms, \MaxLPA{} has a message complexity of $\Omega(m)$, with $m$ the number of edges in the graph that is at least $n^{\nicefrac{7}{4}}$;
instead, the message complexity of the \TwoChoices{} dynamics is $\bigO(n \log n)$ regardless of the actual density of the edges on the graph, since the local update rule only looks at 2 labels.
Our algorithm performs an implicit \emph{sparsification} of the graph, an interesting property for the design of sparse clustering algorithms~\cite{sun_distributed_2017}, in particular for opportunistic network settings~\cite{becchetti_average_2017}. 

\subsection{Evolutionary dynamics}

\begin{figure*}[th!]
\centering
\includegraphics[width=0.95\textwidth]{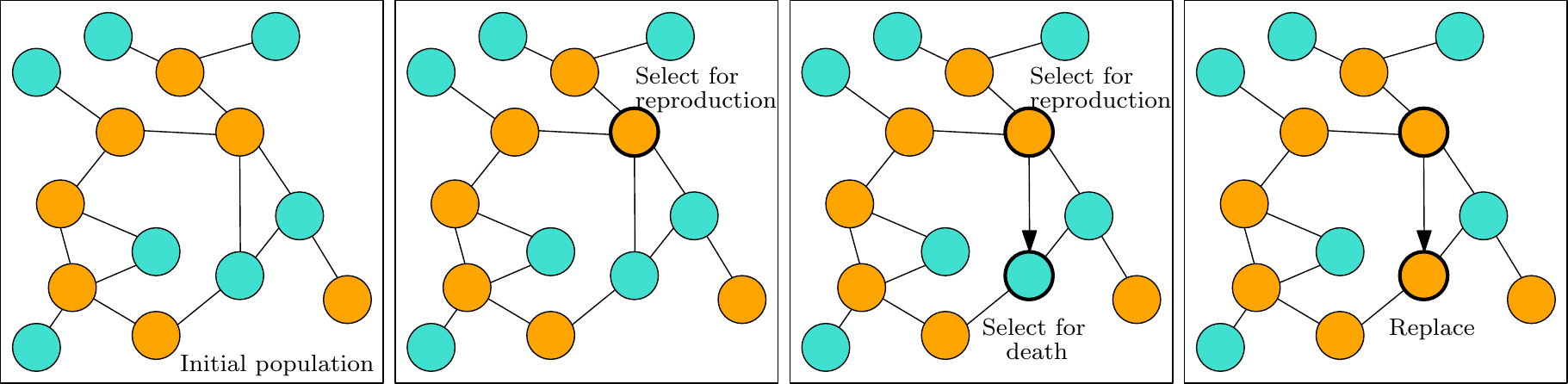}
\caption{Visual representation of the \emph{Moran process} (adapted from \cite{lieberman_evolutionary_2005}).
At each time step an individual is randomly chosen for reproduction according to its \emph{fitness}, 
and a second  individual adjacent to it is randomly chosen for death; 
the offspring of the first individual then replaces the second. 
When the underlying network is regular, the process is equivalent to the \voter{} dynamics~\cite{berenbrink_bounds_2016}. 
}
\label{fig:moran}
\end{figure*}

\emph{Evolutionary dynamics} is the branch of genetics which studies how populations evolve genetically as a result of the interactions among the individuals~\cite{durrett_by_2011}. 
The study of evolutionary dynamics on graphs started with the investigation of the \emph{fixation probability} of the \emph{Moran process} (Figure~\ref{fig:moran}) on different families of graphs, namely the probability that a new mutation with increased fitness eventually spreads across all individuals in the population~\cite{lieberman_evolutionary_2005}. 
The Moran process has since then attracted the attention of the computer science community due to the algorithmic questions associated to its fixation probability~\cite{giakkoupis_amplifiers_2016,galanis_amplifiers_2017}. 

However, no simple dynamics has been proposed so far in the context of evolutionary graph theory for explaining one of evolution's fundamental phenomena, namely \emph{speciation}~\cite{results_speciation_2004}. 
Two fundamental classes of driving forces for speciation can be distinguished: \emph{allopatric speciation} and \emph{sympatric/parapatric speciation}. 
The former, which refers to the divergence of species resulting from geographical isolation, is nowadays considered relatively well understood~\cite{savolainen_sympatric_2006}; 
on the contrary, the latter, namely divergence without complete geographical isolation, is still controversial~\cite{savolainen_sympatric_2006,bolnick_sympatric_2007}.
In several evolutionary settings the spread of a mutation appears nonlinear with respect to the number of interacting individuals carrying the mutation, exhibiting a drift towards the most frequent phenotypes~\cite{results_speciation_2004}. 
In this work we look at the \TwoChoices{} dynamics as a quadratic evolutionary dynamics on a clustered graph representing sympatric and parapatric scenarios.
We regard the random initialization of the \TwoChoices{} process as two inter-mixed populations of individuals with different genetic pools. 
The interactions for reproduction purposes between the two populations can be categorized in frequent interactions among individuals within an equal-size bipartition of the populations, i.e., the \emph{communities}, and less frequent interactions between these two communities which, in later stages of the differentiation process, may be interpreted as genetic admixture, i.e. interbreeding between two genetically-diverging populations~\cite{martin_genome-wide_2013}.

Within the aforementioned framework our Theorem~\ref{theorem:main} provides an analytical evolutionary graph-theoretic proof of concept on how speciation can emerge from the simple nonlinear underlying dynamics of the evolutionary process at the population level. 

\subsection{Computational dynamics}

Dynamics are rules to update an agent's state according to a function which is invariant with respect to time, network topology, and identity of an agent's neighbors, and whose arguments are only the agent's current state and those of its neighbors~\cite{mossel_opinion_2017,emanuele_natale_computational_2017}.
Simple models of interaction between pairs of nodes in a network have been studied since the first half of the 20th century in statistical mechanics~\cite{liggett_interacting_2012} and in the second half in diverse sciences, such as economics and sociology, where averaging-based opinion dynamics such as the DeGroot model have been investigated~\cite{french_formal_1956,degroot_reaching_1974,jackson_social_2010}.
The first study in computer science of a dynamics from a computational point of view is that of a synchronous-time version of the \voter{} dynamics, where, in each discrete-time round, each node looks at a random neighbor and copies its opinion~\cite{hassin_distributed_2001}. 
The \voter{} dynamics can be regarded as the simplest dynamics, in the sense that there is arguably no simpler rule by which nodes may meaningfully update their state as a function of their neighbors' states.
Examples of other dynamics are: \textsc{Undecided-State}~\cite{clementi_tight_2017}, \textsc{3-Majority}~\cite{becchetti_simple_2017,becchetti_stabilizing_2016}, \textsc{2-Median}~\cite{Doerr:2011}, \textsc{Averaging}.
The \textsc{Averaging} dynamics has been employed for solving the Community Detection task~\cite{becchetti_find_2017}. 
However, we remark that the resulting protocol is not classifiable within LPAs: The configuration space in not described in terms of the finite set of labels initially used by nodes, but by rational values generated from the averaging update rule.
Other examples of problems for which dynamics have been successfully employed in order to design an efficient solution are Noisy Rumor Spreading~\cite{fraigniaud_noisy_2016}, Exact Majority~\cite{mertzios2017determining}, and Clock Synchronization~\cite{boczkowski_minimizing_2017}.

We now focus on the \TwoChoices{} dynamics, which is the subject of this work. 
It can arguably be considered the simplest type of dynamics after the \voter{} dynamics and, until now, it constitutes one of the few processes whose behavior has been characterized on non-complete topologies~\cite{cooper_power_2014,cooper_fast_2015,cooper_fast_2016}.
It has been proven that a network of agents, each with a binary state, will support the initially most frequent opinion with high probability after a polylogarithmic number of rounds whenever the initial \emph{bias} (the advantage of a state on the other) is greater than a function of the network's \emph{expansion}~\cite{cooper_power_2014}.
Such result was later refined with milder assumptions on the initial bias with respect to the network's expansion~\cite{cooper_fast_2015} and generalized to more opinions~\cite{cooper_fast_2016}.
Moreover, in core-periphery networks, depending on the strength of the cut between the core and the periphery, a phase-transition phenomenon occurs~\cite{AAMAS2018}: 
Either one of the colors rapidly spreads over the rest of the network, or a metastable phase takes place, in which both the colors coexist in the network for superpolynomial time.

\section{Notation}

Let $G=(V, E)$ be a $(2n,d,b)$-\emph{clustered regular graph} (Definition~\ref{def:clustered-regular-graph}) and let us define $a := d-b$.
Note that $G$ is composed by two $a$-regular communities connected by a $b$-regular cut (Figure~\ref{fig:regularclustered}) and that when $a > b$ the graph $G$ exhibits a well-clustered structure, i.e., each node has more neighbors in its community than in the other one.
\begin{figure}
    \centering
    \includegraphics[width=0.95\columnwidth]{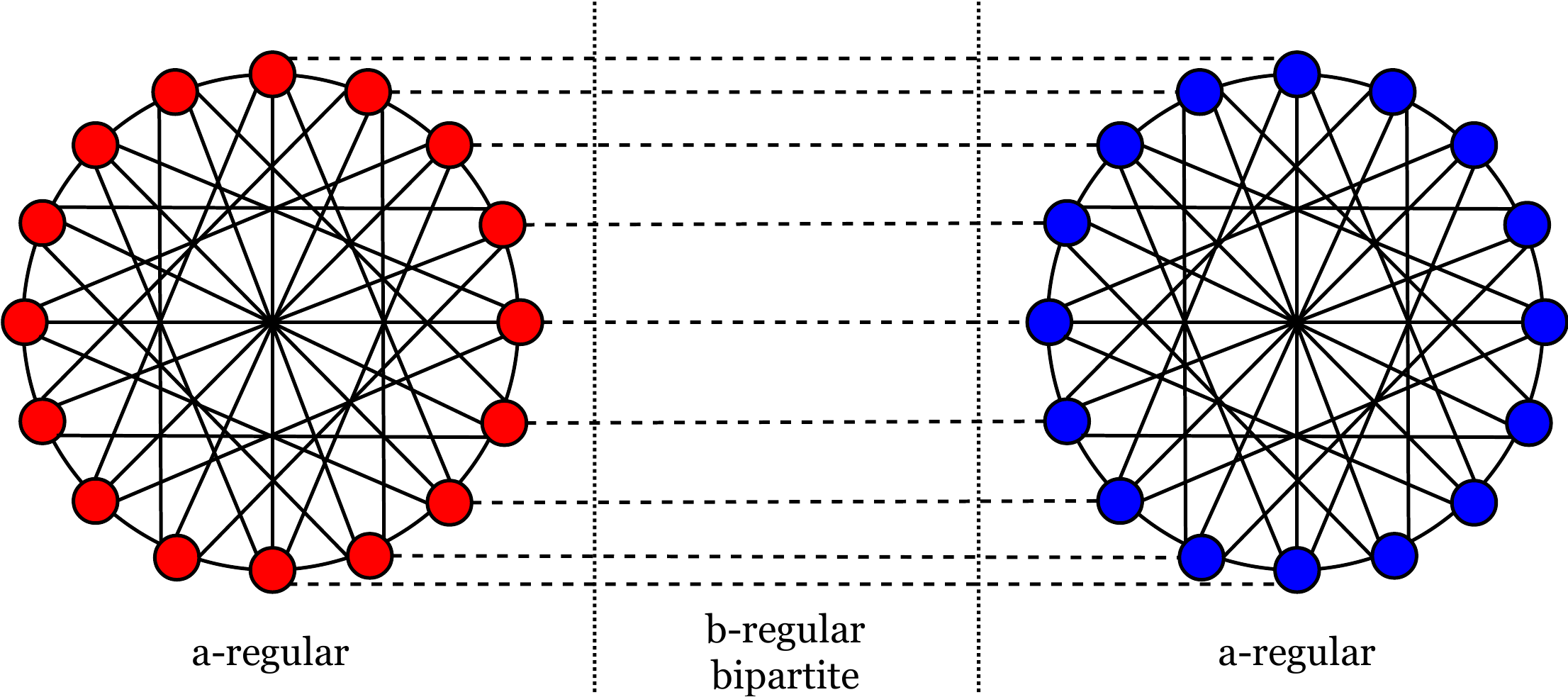}
    \caption{Representation of a $(2n,d,b)$-clustered regular graph where $a:=d-b$. Each community induces an $a$-regular graph while the cut between the two communities induces a $b$-regular bipartite graph.}
    \label{fig:regularclustered}
\end{figure}
\begin{definition}[Clustered regular graph]
    \label{def:clustered-regular-graph}
    A $(2n, d, b)$-clustered regular graph~\cite{becchetti_find_2017} is a graph $G = (V, E)$ such that:
    \begin{itemize}
      \item $V = V_1 \cup V_2$, $V_1 \cap V_2 = \emptyset$, and $|V_1| = |V_2| = n$;
      \item every node has degree $d$; 
      \item every node in $V_1$ has exactly $b$ neighbors in $V_2$ and every node in $V_2$ has exactly $b$ neighbors in $V_1$.
    \end{itemize}
\end{definition}

Each node of $G$ maintains a binary \emph{state} that we represent as a color: either \emph{red} or \emph{blue}.
We denote the vector of states of all nodes in $G$ at time $t$ as the \emph{configuration} vector $\bc{t}$
and we refer to the state of a node $u \in V$ at time $t$ as $\bc{t}_u \in \{red,\, blue\}$.
We call $B^{(t)}$ the set of nodes colored \emph{blue} at time $t$ and $R^{(t)}$ the set of nodes colored \emph{red} at time $t$. 
For each community $i \in \{1,2\}$ we define $B^{(t)}_i := V_i \cap B^{(t)}$ and $R^{(t)}_i := V_i \cap R^{(t)}$. 
We call $s^{(t)}_i = \vert R^{(t)}_i \vert - \vert B^{(t)}_i \vert$ the \emph{bias} in community $i$ toward color \emph{red}.
Given some initial configuration $\bc{0}$, we let the nodes of $G$ run the following \TwoChoices{} dynamics.

\begin{definition}[2-Choices dynamics]
    \label{def:2choices}
    The \TwoChoices{} dynamics is a local synchronous protocol that works as follows: In each round, each node $u$ chooses two neighbors $v, w$ uniformly at random with replacement; if $v$ and $w$ support the same color, then $u$ updates its own color to their color, otherwise $u$ keeps its previously supported color.
\end{definition}

Note that the random sequence of configurations $\{\bc{t}\}_{t \in \mathbb{N}}$ generated by multiple iterations of the \TwoChoices{} dynamics on $G$ is a Markov Chain with two absorbing states, namely the configurations where all the nodes support the same color, either \emph{red} or \emph{blue}.

Let us now introduce the notion of \emph{almost-clustered configuration}.
\begin{definition}[Almost-clustered configuration]
    \label{def:almost-clustered-config}
    A configuration $\bc{t}$ is \emph{almost-clustered} if 
    \[\textstyle
        \vert s_i \vert  \geq n - \bigO \left(\frac{\log n}{\log \log n}\right)
    \]
    for each $i \in \{1,2\}$ and the sign of the biases is different, i.e., $s_1 s_2 < 0$.
\end{definition}

Intuitively, \emph{almost-clustered} configurations are such that the vast majority of the nodes in one community is supporting one of the two colors, and the vast majority of nodes in the other community is supporting the other color.

In the rest of the section we introduce the notation used to describe the spectral properties of the transition matrix of the underlying graph $G$: 
The analysis in expectation of the process (Lemma~\ref{lemma:expectation}) exploits such spectral properties and
our main result (Theorem~\ref{theorem:main}) makes assumptions on the spectrum of the transition matrix of $G$.

Let $P=\frac{1}{d}A$ be the transition matrix of a simple random walk on $G$, where we denote with $d$ the degree of the nodes and with $A$ the adjacency matrix of $G$.
Note that the transition matrix $P$ can be decomposed as follows:
\begin{multline*}\textstyle
P =
\left(
\begin{array}{cc}
    P_{1,1} & P_{1,2} \\
    P_{2,1} & P_{2,2} 
\end{array}
\right)
= A + B = 
\left(
\begin{array}{cc}
    P_{1,1} & 0 \\
    0 & P_{2,2} 
\end{array}
\right)
+ \left(
\begin{array}{cc}
    0 & P_{1,2} \\
    P_{2,1} & 0
\end{array}
\right),
\end{multline*}
where $A$ is the transition matrix of the communities if we disconnect them,
while $B$ is the transition matrix of the bipartite graph connecting the two communities.
Note that since the cut is regular $B$ is symmetric and $P_{1,2}^\intercal = P_{2,1}$.

We denote with $\lambda_1 \geq \ldots \geq \lambda_n$ the eigenvalues of the transition matrix of the subgraph induced by the first community $\bar{P}_{1,1} := \frac{d}{a} P_{1,1}$ and with $\bv_1, \ldots, \bv_n$ their corresponding eigenvectors;
we denote with $\mu_1 \geq \ldots \geq \mu_n$ the eigenvalues of the transition matrix of the subgraph induced by the second community $\bar{P}_{2,2} := \frac{d}{a} P_{2,2}$ and with $\bw_1, \ldots, \bw_n$ their corresponding eigenvectors.
Since both $\bar{P}_{1,1}$ and $\bar{P}_{2,2}$ are stochastic matrices we have that 
$\lambda_1 = \mu_1 = 1$ and that $\bv_1 = \bw_1 = \frac{1}{\sqrt{n}} \bone$, where $\bone$ is the vector of all ones.
We consider the case in which both the subgraphs induced by the communities are connected and not bipartite; 
thus it holds that $\lambda_2 < 1$, $\mu_2 < 1$ and that $\lambda_n > -1$, $\mu_n > -1$.

We define $\lambda := \max(|\lambda_2|, |\lambda_n|, |\mu_2|, |\mu_n|)$. 
The value of $\lambda$ is a representative of the second largest eigenvalues for both the subgraphs induced by the communities and is closely related to the third largest eigenvalue of $P$.

In addition to the analysis in expectation, we also provide concentration bounds for the behavior of the process.
In this context, we say that an event $\mathcal{E}$ happens \emph{with high probability} (for short, \emph{w.h.p.}) if $\Prob{}{\mathcal{E}} \geq 1 - \bigO(n^{-\gamma})$, for some constant $\gamma>0$.

\section{Analysis of the \TwoChoices{} dynamics}
\label{sec:analysis}

In this section we give a high-level overview of the main steps and ideas used for the analysis of the process.

Let $G$ be a clustered regular graph (Definition~\ref{def:clustered-regular-graph}). 
Let each node in $G$ initially pick a color $\bc{0}_u \in \{red,\, blue\}$ uniformly at random and independently from the other nodes. 
Then let the nodes of $G$ run the \TwoChoices{} dynamics (Definition~\ref{def:2choices}).

The variance in the initialization suggests that with some constant probability the distribution of the two colors will be slightly asymmetric w.r.t.\ the two communities, i.e., the first community will have a bias toward a color, while the second community will have a bias toward the other color.
Without loss of generality, we consider the case in which $s_1$ is positive and $s_2$ is negative, i.e., the first community is unbalanced toward color \emph{red} while the second community is unbalanced toward color \emph{blue}. 

Roughly speaking, we show that when the initialization is ``lucky'', i.e., the biases in the two communities are toward different colors, there is a significant probability that the process will rapidly make the distribution more and more asymmetric until converging to an \emph{almost-clustered} configuration (Definition~\ref{def:almost-clustered-config}), i.e., a configuration in which, apart from a small number of outliers, the nodes in the two communities support different colors.
This behavior of the \TwoChoices{} dynamics is formalized in the following theorem.
\begin{theorem}[Constant probability of clustering]
    \label{theorem:main}
    Let $G = (V, E)$ be a connected $(2n,d,b)$-clustered regular graph such that
    $\frac{b}{d} = \bigO(n^{-1/2})$ and $\lambda = \bigO(n^{-1/4})$.
    Let $c \in \mathbb{N}$ be any constant;
    let us define the two following events about the \TwoChoices{} dynamics on $G$:
    \begin{quote}
    $\xi$: ``Starting from a \textit{random initialization} the process reaches an \emph{almost-clustered} configuration within $\bigO(\log n)$ rounds.''
    \\
    $\xi_c$: ``Starting from an \emph{almost-clustered} configuration the process stays in \emph{almost-clustered} configurations for $n^c$ rounds.''
    \end{quote}
    For two suitable positive constants $\gamma_1$ and $\gamma_2$ it holds that
    \[
        \Prob{}{\xi} \geq \gamma_1 
        \,\text{ and }\,
        \Prob{}{\xi_c} \geq 1 - n^{-\gamma_2}.
    \]
\end{theorem}

\begin{proof}
The proof is divided in the following steps:
\begin{enumerate}
    \item The bias in each community is initially $\vert s_i \vert = \Theta(\sqrt{n})$, for each $i \in \{1,2\}$, and the sign of the biases is different, with constant probability (Lemma~\ref{lemma:init});
    
    \item The bias in each community becomes $\vert s_i \vert = \Theta(\sqrt{n} \log n)$, for each $i \in \{1,2\}$, in $\bigO(\log \log n)$ rounds and the sign of the biases is preserved, with constant probability (Lemma~\ref{lemma:clustered-breaking-symmetry});
    
    \item The bias in each community becomes $|s_i| \geq n - \bigO(\log n)$, for each $i \in \{1,2\}$, in $\bigO(\log n)$ rounds and the sign of the biases is preserved, with high probability (Lemma~\ref{lemma:global-convergence});
    
    \item The process enters an \emph{almost-clustered} configuration in one single round and lies in the set of \emph{almost-clustered} configurations for the next $n^c$ rounds, with high probability
    (Lemma~\ref{lemma:metastability}).
\end{enumerate}

The proofs of the lemmas used in Theorem~\ref{theorem:main} are deferred to the Appendix.

Before starting with the proof, let us introduce some extra notation.
Let $\frac{b}{d}\leq c_1 \cdot n^{-1/2}$ for some positive constant $c_1$, i.e., let every node in each community have at most $c_1$ neighbors in the opposite community for every $\sqrt{n}$ neighbors in their own.
Let $\lambda \leq c_2 \cdot n^{-1/4}$, for some positive constant $c_2$;
note that the hypothesis on $\lambda$ implies that the subgraph induced by each community is a good expander.
Let us define the constant $h := 4(2\sqrt{2}c_1 + c_2^2)$.

We start the analysis of the process by looking at the initialization phase.
In particular, in Lemma~\ref{lemma:init} we show that there is a probability at least constant that the initialization is ``lucky'', i.e., that the biases in the two communities are $\Theta(\sqrt{n})$ toward different colors. 
This is true because the Binomial distribution, i.e., the initial distribution of the colors in the graph, is well approximated by a Gaussian distribution, and the latter has a constant probability to deviate from the mean by the standard deviation.
The Central Limit Theorem establishes the approximation of the distribution and we are able to quantify it using the Berry-Esseen Theorem.

\begin{restatable}[Lucky initialization]{lem}{lemmainit}
\label{lemma:init}
Let $G=(V, E)$ be a $(2n, d, b)$-clustered regular graph and let each node $u \in V$ choose a color $\bc{0}_u \in \{red,\, blue\}$ uniformly at random and independently from the others. 
Let $c_1$ and $c_2$ be two positive constants. 
Then, there exists a constant $\gamma_1$ such that
\[
\Prob{}{s^{(0)}_1 \geq h \sqrt{n} \wedge  - s^{(0)}_2 \geq h \sqrt{n}} \geq \gamma_1.
\]
\end{restatable}

Then, considering a configuration $\bc{t}$ at a generic time $t$, we look at the expected evolution of the process observing the behavior of one single community, but also taking into account the influence of the other.
Informally, Lemma~\ref{lemma:expectation} gives a bound to the number of nodes that will support the minority color in each community at the next round as a function of all the parameters involved in the process: 
the number of nodes supporting the minority color in each  community at the current round; 
the number of nodes supporting the same color in the other community at the current round; 
the expansion of the communities $\lambda \leq c_2 \cdot n^{-1/4}$; 
the cut density $\frac{b}{d} \leq c_1 \cdot n^{-1/2}$.

The proof of Lemma~\ref{lemma:expectation} leverages the fact that the expected number of nodes supporting a given color can be expressed as a quadratic form of the transition matrix of a simple random walk on the graph, allowing to relate the behavior of the process to the expansion of the communities, as exploited in~\cite{cooper_fast_2015,cooper_fast_2016}.
 
\begin{restatable}[Expected decrease of the minority color]{lem}{lemmaexpectation}
\label{lemma:expectation}
Let $G$ be a $(2n, d, b)$-clustered regular graph. 
For any configuration $\bc{t}$ we have that
\begin{multline*}\textstyle
\Ex{|B_1^{(t+1)}| \cond \bc{t}} 
< 
|B_1^{(t)}| \left[
    1 - \frac{s_1}{2n} + \frac{c_2^2}{\sqrt{n}} + 
\right.
\\\textstyle
\hspace{1cm}
\left.
    + \frac{2c_1}{\sqrt{n}} \sqrt{\frac{|B_2^{(t)}|}{|B_1^{(t)}|}\left( \frac{1}{2} - \frac{s_1}{2n} + \frac{c_2^2}{\sqrt{n}} + \frac{c_1^2 |B_2^{(t)}|}{n |B_1^{(t)}|}\right)}
\right]
\end{multline*}
and
\begin{multline*}\textstyle
\Ex{|R_2^{(t+1)}| \cond \bc{t}} 
< 
|R_2^{(t)}| \left[
    1 + \frac{s_2}{2n} + \frac{c_2^2}{\sqrt{n}} +
\right.
\\\textstyle
\hspace{1cm}
\left.
    + \frac{2c_1}{\sqrt{n}} \sqrt{\frac{|R_1^{(t)}|}{|R_2^{(t)}|}\left( \frac{1}{2} + \frac{s_2}{2n} + \frac{c_2^2}{\sqrt{n}} + \frac{c_1^2 |R_1^{(t)}|}{n |R_2^{(t)}|}\right)}
\right].
\end{multline*}
\end{restatable}

It follows from Lemma~\ref{lemma:expectation} that the asymmetry in the coloring of the nodes in the two communities continues to grow in expectation.
In fact, when in a certain range of values, the bias in the first community increases in expectation at each round while the bias in the second community decreases in expectation at each round, since the minority color in each community decreases.
With Lemma~\ref{lemma:unconditional_bias_increase} we prove that the increase of the bias in the first community and the decrease of the bias in the second community we have shown in expectation in Lemma~\ref{lemma:expectation} is multiplicative w.h.p.\ whenever $s_1$ satisfies $s_1 \in [h\sqrt{n}, \frac{n}{2}]$ and $s_2$ satisfies $s_2 \in [-\frac{n}{2}, -h\sqrt{n}]$.
With the use of concentration of probability arguments, namely a multiplicative form of the Chernoff bounds~\cite[Lemma~1.1]{dubhashi2009concentration}, we show that the number of nodes with the minority color in each community decreases and we use this fact to prove Lemma~\ref{lemma:unconditional_bias_increase}.

\begin{restatable}[Probability of multiplicative growth of the bias]{lem}{lemmabiasincrease}
\label{lemma:unconditional_bias_increase}
Let $\bc{t}$ be a configuration such that 
$h\sqrt{n} \leq s_1 \leq \frac{n}{2}$
and $h\sqrt{n} \leq -s_2 \leq \frac{n}{2}$.
Then, it holds that
\[
\Prob{}{s_1^{(t+1)} \geq (1 + 1/16)\, s_1 \cond \bc{t}} 
\geq 1 - e^{-\nicefrac{2s_1^2}{32^2 n}}
\]
and
\[
\Prob{}{s_2^{(t+1)} \leq (1 + 1/16)\, s_2 \cond \bc{t}} 
\geq 1 - e^{-\nicefrac{2s_2^2}{32^2 n}}.
\]
\end{restatable}

Now we know that there is a constant probability that the initialization of the process starts is ``lucky'' (Lemma~\ref{lemma:init}); 
we also know that the bias in the first community will increase in expectation and the bias in the second community will decrease in expectation (Lemma~\ref{lemma:expectation});
moreover, when in a given range, we know that the biases will follow their expected behavior with high probability (Lemma~\ref{lemma:unconditional_bias_increase}).

Then we need to show that the asymmetry in the coloring of the two communities will rapidly increase up to a configuration such that $|s_i| = \Theta({\sqrt{n}\log n})$, for each $i \in \{1,2\}$, while the sign of the biases is preserved.
More formally, with Lemma~\ref{lemma:clustered-breaking-symmetry} we prove the internal symmetry breaking of each community. 
This is possible by applying Lemma~\ref{lemma:init}, and by iterating the application of Lemma~\ref{lemma:unconditional_bias_increase} for $\bigO(\log \log n)$ rounds, i.e., until the bias is large enough; 
finally we handle the stochastic dependency between the two biases during their respective increases in opposite directions.

\begin{restatable}[Clustering -- Symmetry Breaking]{lem}{lemmaclusteringbreakingsymmetry}
\label{lemma:clustered-breaking-symmetry}
Starting from an initial configuration where each node $u \in V$ chooses a color $\bc{0}_u \in \{red, blue\}$ uniformly at random and independently from the others, it holds that, with constant probability, within $\bigO(\log \log n)$ rounds  the process reaches a configuration $\bc{t}$ such that 
\[
    s^{(t)}_1 \geq \sqrt{n} \log n
    \,\text{ and }\,
    -s^{(t)}_2 \geq \sqrt{n} \log n.
\]
\end{restatable}

Once the internal symmetry of each community is broken, we show that, with high probability, both biases keep increasing while preserving their sign until they rapidly reach a configuration in which the minority color in each community has at most logarithmic size.
This behavior is formally proved in Lemma~\ref{lemma:global-convergence}, 
again through the application of Lemma~\ref{lemma:expectation} and Lemma~\ref{lemma:unconditional_bias_increase}.

\begin{restatable}[Convergence]{lem}{lemmaconvergence}
\label{lemma:global-convergence}
Starting from a configuration $\bc{t}$ such that $\vert s_i \vert \geq \sqrt{n} \log n$, for each $i \in \{1,2\}$, 
there exist two rounds $\tau_1,\tau_2 = \bigO(\log n)$ such that 
\[
    \vert s^{(\tau_1)}_1 \vert \geq n - \log n
    \,\text{ and }\,
    \vert s^{(\tau_2)}_2 \vert \geq n - \log n
\]
and the sign of the biases is preserved, with high probability.
\end{restatable}

Finally, with Lemma~\ref{lemma:metastability} we show that the number of wrongly colored nodes in each community drops to $\bigO(\log n / \log \log n)$ in one single round (by approximating it with a Poisson random variable through an application of Le~Cam's Theorem) and then, with high probability, the process enters a \emph{metastable} phase in which the only possible configurations are \emph{almost-clustered}; this will last for any polynomial number of rounds.
In other words, even if a few nodes in each community will continue to change color, almost all the nodes in one community will support one color while almost all the nodes in the other community will support the other color.
Note that this quantity is \emph{tight}: It is possible to prove that, within any polynomial number of rounds, there will be a round in which at least $\Omega(\log n / \log\log n)$ nodes in each community will have the wrong color.

\begin{restatable}[Metastability]{lem}{lemmametastability}
\label{lemma:metastability}
Let $c \in \mathbb{N}$ be any constant.
Starting from a configuration $\bc{t}$ such that $\vert s_i \vert \geq n - \log n$ for each $i \in \{1,2\}$, for the next $n^c$ rounds the process lies in the set of configurations such that 
\[\textstyle
    \vert s_i \vert \geq n - \bigO\left(\frac{\log n}{\log \log n}\right)
\]
and the sign of the bias is preserved, with high probability.
\end{restatable}

More formally, through Lemma~\ref{lemma:global-convergence} and Lemma~\ref{lemma:metastability} we can finally prove that $\Prob{}{\xi} \geq \gamma_1$ and $\Prob{}{\xi_c} \geq 1 - n^{-\gamma_2}$ for any constant $c$, concluding the proof of Theorem~\ref{theorem:main}.
\end{proof}

\section{Distributed Label Propagation Algorithm via\\ Community-Sensitive Labeling}
\label{sec:csl}

We showed that, starting from a random initialization, the \TwoChoices{} dynamics reaches an \emph{almost-clustered} configuration within $\bigO(\log n)$ rounds with constant probability. 
This result is tight, given that there is constant probability that the two communities converge to the same color. 
Similarly to Lemma~\ref{lemma:init}, it holds that with constant probability both the biases are unbalanced toward the same color, i.e., $s^{(0)}_1 \geq h\sqrt{n}$ and $s^{(0)}_2 \geq  h\sqrt{n}$. It means that a suitable variant of Lemma~\ref{lemma:clustered-breaking-symmetry} shows that there is constant probability that within $\bigO(\log \log n)$ rounds the process reaches a configuration such that $ s^{(t)}_1 \geq \sqrt{n} \log n \,\text{ and }\, s^{(t)}_2 \geq \sqrt{n} \log n$.
Then, Lemma~\ref{lemma:global-convergence} and Lemma~\ref{lemma:metastability} show that the system gets quickly stuck in a configuration where almost all nodes have the same color.
This is a proof that, given the symmetric nature of the process, we need some luck in the initialization to reach an \emph{almost-clustered} configuration.

In order to get an algorithm that works w.h.p.\ we sketch how to use the results of the previous sections to build a Community-Sensitive Labeling~\cite{becchetti_average_2017} within $\Theta(\log n)$ rounds.
A Community-Sensitive Labeling (CSL) is made up by a labeling of the nodes and a predicate that can be applied to pairs of labels;
it holds that, for all but a small number of outliers, the predicate is satisfied if the nodes belong to the same community, and it is not satisfied if the nodes belong to different communities.

\begin{theorem}[LPA via CSL]
Let $G = (V, E)$ be a connected and nonbipartite $(2n,d,b)$-clustered regular graph such that
$\frac{b}{d} = \bigO(n^{-1/2})$ and $\lambda = \bigO(n^{-1/4})$.
Let $\bc{0}$ be the initial configuration, where each node $u \in V$ picks a \emph{vector} of colors $\bc{0}_u \in \{red, blue\}^{\ell}$ sampled uniformly at random and independently from the other nodes, such that $\ell = c\log n$ for some positive constant $c$.
Consider the resulting vector after $\Theta(\log n)$ rounds of independent parallel runs of the \TwoChoices{} dynamics, each one working on a different component of the vector: 
For all the pairs of nodes but a polylogarithmic number, it holds that the vectors of nodes in the same community are equal while the vectors of nodes in different communities are different.
\end{theorem}

\begin{proof}[Sketch of proof]
As for the first part of the predicate, it is a simple application of Theorem~\ref{theorem:main}. Indeed, at least one of the $\Theta(\log n)$ runs of the \TwoChoices{} dynamics ends in an \emph{almost-clustered} configuration with probability $1 - \gamma^{-\Theta(\log n)} = 1 - n^{-\Theta(1)}$. 
As for the second part we show that no matter if the process reaches an almost-clustering, nodes in the same community will have the same color with high probability.
This is consequence of Lemma~\ref{lemma:global-convergence} and of the following one, which we can prove by applying a general tool for Markov Chains~\cite[Lemma~4.5]{clementi_tight_2017}.

\begin{restatable}[Consensus -- Symmetry Breaking]{lem}{lemmaconsensusbreakingsymmetry}
\label{lemma:consensus-breaking-symmetry}
Starting from any initial configuration $\bc{0}$, within $\bigO(\log n)$ rounds the system reaches a configuration
$\bc{t}$ such that 
\[
    \vert s^{(t)}_1 \vert \geq \sqrt{n} \log n
    \,\text{ and }\,
    \vert s^{(t)}_2 \vert \geq \sqrt{n} \log n,
\]
with high probability.
\end{restatable}

Thus, most pairs of nodes can locally distinguish if they are in the same community with high probability by checking whether their vectors differ on any component.
\end{proof}

\section{Conclusions and future work}

We focused on providing a proof of concept of how spectral techniques and concentration of probability results can be combined to provide a rigorous analysis of the behavior of dynamics converging to metastable configurations that reflect structural properties of the network. 
In turns, we identified two important implications of our result, which we discussed in the Introduction and we briefly recall here.
In the context of \emph{graph clustering}, it constitute the first analytical result on a distributed \emph{label propagation algorithm} with quasi-linear message complexity, contributing to a deeper understanding of such class of widely applied heuristics to detect communities in networks.
In the framework of \emph{evolutionary biology}, it provides a simplistic model of how \emph{species differentiation} may occur as the result of the interplay between the local interaction rule at the population level and the underlying topology that describes such interaction. 

A limitation of our approach is the restriction to regular topologies. 
The regularity assumption greatly simplifies the calculations, which are still quite involved. 
However, it has been shown in~\cite{cooper_fast_2015} that a similar analysis can be performed for general topologies. 
Thus, it should be possible to extend our analysis to the irregular case, at the price of a much greater amount of technicalities. 
For example, it should be possible to prove a generalization of our result to the class of $(2n, d, b, \gamma)$-clustered graphs investigated in~\cite{becchetti_find_2017}, which relaxes the class of $(2n, d, b)$-clustered graphs by assuming that each node has $d\pm \gamma d$ neighbors of which $b\pm \gamma d$ belongs to the other community. 
In fact it is possible to bound the second eigenvalue of the graph in a way which approximates (depending on $\gamma$) the $(2n, d, b)$-clustered graphs case considered here using~\cite[Lemma C.2]{becchetti_find_2017}.
Another important issue is to get a denser cut, at least parametrized w.r.t.\ the number of edges inside each community. This cannot be achieved by slightly changing the analysis of this paper, but requires a different approach, since it is possible to show that the technique used in Lemma~\ref{lemma:expectation} brings to a sparse cut.
Finally, an interesting direction is the use of domination arguments, perhaps based on coupling techniques, to generalize our result to more general dynamics which interpolates between the \emph{quadratic} \TwoChoices{} dynamics and the \emph{linear} \textsc{Voter} dynamics~\cite{berenbrink_ignore_2017}.
In particular, this latter direction would have more general implications in the practical contexts discussed in this work, namely label propagation algorithms and evolutionary dynamics.

% references
\bibliographystyle{alpha}
\bibliography{references}

\clearpage
\appendix
\section{Omitted Proofs of Section~\ref{sec:analysis}}

\lemmainit*
\begin{proof}
The initial bias of the first cluster $s_1 = |R_1| - |B_1|$ can be thought as a sum of Rademacher random variables,
i.e.\ $s_1 = \sum_{i \in V_1} X_i$ where $X_i = 1$ if node $i$ chooses color \emph{red} and $X_i = -1$ if node $i$ chooses color \emph{blue}.
Rademacher random variables have mean equal to 0, variance equal to 1, and third moment equal to 1; 
thus, we can apply the Berry-Esseen theorem (Theorem~\ref{thm:berry-esseen})
which in our case states that
\[
\left\vert 
    \Prob{}{\frac{\sum_{i \in V_1}{X_i}}{\sqrt{n}} \leq h} - \Phi(h)
\right\vert
    \leq \frac{C}{\sqrt{n}},
\]
where $\Phi$ is the cumulative distribution function of the standard normal distribution and $C$ is a universal positive constant.
Hence, 
\begin{align*}
\Prob{}{\sum_{i \in V_1}{X_i} < 4(2\sqrt{2}c_1 + c_2^2) \sqrt{n}}
&\leq \Phi(4(2\sqrt{2}c_1 + c_2^2)) + \frac{C}{\sqrt{n}}\\
&\stackrel{(a)}{\leq} \Phi(4(2\sqrt{2}c_1 + c_2^2)) + \epsilon \label{eq:epsilon}
\stackrel{(b)}{\leq} \alpha,
\end{align*}
where in $(a)$ $\epsilon$ is a suitably small positive constant and $(b)$ holds for a positive constant $\alpha$ strictly smaller than one, because for every $h$ constant also $\Phi(h)$ is a constant strictly smaller than one.
The same inequality also holds for $s_2$. 
Note that, since the random variables $s_1$ and $s_2$ are independent, we have:
\[
\Prob{}{s_1 \geq h \sqrt{n}} \cdot \Prob{}{s_2 \geq h \sqrt{n}}
\geq (1-\alpha)^2 = \gamma_1.
\]
\end{proof}

\lemmaexpectation*
\begin{proof}
W.l.o.g.\ we analyze the case of the \emph{blue} minority color in community $1$. 
The proof is completely symmetric for the \emph{red} minority color in community $2$.

For every set $Z \in \{B,B_1,B_2,R,R_1,R_2\}$ and for every node $v \in V$, we define $Z(v) = N(v) \cap Z$, where $N(v)$ is the set of neighbors of $v$. 
Thus, by definition of \TwoChoices{} dynamics (Definition~\ref{def:2choices}),
we can write the expected number of nodes supporting the minority color in community $1$ at round $t+1$
as the sum of the probabilities for each node supporting color \emph{red} of picking two \emph{blue} nodes (and thus becoming \emph{blue}) 
and the sum of the probabilities for each \emph{blue} node of not picking two \emph{red} nodes (and thus remaining \emph{blue}). 
\begin{align*}
&\Ex{|B_1^{(t+1)|} \cond \bc{t}} =
\sum_{x \in R_1} \left( \frac{|B(x)|}{d} \right)^2 + \sum_{x \in B_1}\left( 1 - \left( \frac{|R(x)|}{d} \right)^2\right)\nonumber  \\
 &=
\sum_{x \in V_1} \left( \frac{|B(x)|}{d} \right)^2 - \sum_{x \in B_1} \left( \frac{|B(x)|}{d} \right)^2
	+ \sum_{x \in B_1} \left(1 - \left(1 - \frac{|B(x)|}{d} \right)^2 \right)\nonumber \\
 &=
\sum_{x \in V_1} \left( \frac{|B(x)|}{d} \right)^2 
	- \sum_{x \in B_1} \left( \frac{|B(x)|}{d} \right)^2 
		+ \sum_{x \in B_1} \left( 1 - 1 + 2\frac{|B(x)|}{d} - \left( \frac{|B(x)|}{d} \right)^2 \right)\nonumber  \\
 &=
\sum_{x \in V_1} \left( \frac{|B(x)|}{d} \right)^2 + 2 \sum_{x \in B_1} \left( \frac{|B(x)|}{d} - \left( \frac{|B(x)|}{d} \right)^2 \right)\nonumber \\
 &=
 \sum_{x \in V_1} \left( \frac{|B(x)|}{d} \right)^2 
	+  2 \sum_{x \in B_1} \left( \frac{|B(x)|}{d} \, \left( 1 - \frac{|B(x)|}{d} \right) \right)\nonumber\\
 &\leq
\sum_{x \in V_1} \left( \frac{|B(x)|}{d} \right)^2 
	+ \left( \frac{B_1}{2} \right),
\end{align*}
where in the last inequality we used the fact that 
$\frac{|B(x)|}{d} \, \big(1 - \frac{|B(x)|}{d}\big) \leq \frac{1}{4}$,
since it is a concave function and its maximum is $\frac{1}{4}$.

In order to bound the quantity $\sum_{x \in V_1} \left( \frac{|B(x)|}{d} \right)^2$ we use the assumptions on the structure of $G$, 
i.e.\ that it is $(2n, d, b)$-clustered, that $\frac{b}{d} \leq c_1 \cdot n^{-1/2}$, and that $\lambda \leq c_2 \cdot n^{-1/4}$.
In particular, we split the quantity into three terms as follows:
\begin{align*}
\sum_{x \in V_1} \left( \frac{|B(x)|}{d} \right)^2 &= 
    \sum_{x \in V_1} \left( \frac{|B_1(x)|}{d} + \frac{|B_2(x)|}{d} \right)^2 \\
    &= \sum_{x \in V_1} \left( \frac{|B_1(x)|}{d} \right)^2 
    + \sum_{x \in V_1} \left( \frac{|B_2(x)|}{d} \right)^2 
    + 2 \sum_{x \in V_1} \frac{|B_1(x)|}{d}\cdot \frac{|B_2(x)|}{d}.
\end{align*}

We upper bound the first of the terms by using $\lambda := \max(|\lambda_2|, |\lambda_n|, |\mu_2|, |\mu_n|)$, which gives a measure of the internal expansion of the graphs induced by the clusters.
Let $\bar{P}_{1,1} := \frac{d}{a} P_{1,1}$ be the transition matrix of the subgraph induced by the first cluster. 
Notice that, since $G$ is $(2n, d, b)$-clustered, the subgraph induced by the first cluster is $a$-regular and thus $\bar{P}_{1,1}$ is symmetric.
Consequently the eigenvectors of $\bar{P}_{1,1}$ form an orthonormal basis of the space.
%
% vectors
Let $\bb^{(t)}$ be the indicator vector of the set $B$, 
i.e.\ $\bb^{(t)}(v) = 1$ if $v \in B^{(t)}$ and $\bb^{(t)}(v) = 0$ otherwise. When clear from the context we will omit the time $t$.
This allows us to write the matrix in its spectral decomposition, i.e.\ $\bar{P}_{1,1} = \sum_{i=1}^{n} \lambda_i \bv_i \bv_i^\intercal$, 
and the indicator vector of the blue nodes in the first cluster as a linear combination of the eigenvectors of $\bar{P}_{1,1}$, 
i.e.\ $\bone_{B_1} = \sum_{i=1}^{n} \alpha_i \bv_i$ 
with $\alpha_i = \langle \bv_i, \bone_{B_1} \rangle$.
Hence, we get that
\begin{align*}
    \sum_{x \in V_1}  \left( \frac{|B_1(x)|}{d} \right)^2 
    &= \norm{P_{1,1} \bb_1}_2^2
    =\bb_1^\intercal P_{1,1}^\intercal \cdot P_{1,1}\bb_1
    \\
    &= \bb^\intercal_1 P_{1,1}^2 \bb_1 
    = \frac{a^2}{d^2} \bb^\intercal_1 \bar{P}_{1,1}^2 \bb_1
    \\
    &\leq \bb^\intercal_1 \bar{P}_{1,1}^2 \bb_1
    = \bb^\intercal_1
    	\cdot \sum_{i=1}^n \lambda_i^2 \bv_i \bv_i^\intercal
    	\cdot \sum_{i=1}^{n} \alpha_i \bv_i
    \\
    &= \bb^\intercal_1
    	\cdot \sum_{i=1}^n \lambda_i^2 \alpha_i \bv_i
    \\
    &= \bb^\intercal_1
    	\cdot \left( 
    		\lambda_1^2 \alpha_1 \bv_1 + 
    		\sum_{i=2}^{n} \lambda_i^2 \alpha_i \bv_i
    	\right)
    \\
    &\leq \bb^\intercal_1
    	\cdot \left( 
    		\lambda_1^2 \alpha_1 \bv_1 + 
    		\lambda^2 \sum_{i=2}^{n} \alpha_i \bv_i
    	\right)
    \\
    &\leq \bb^\intercal_1
    	\cdot \left( 
    		\lambda_1^2 \alpha_1 \bv_1 + 
    		\lambda^2 \sum_{i=1}^{n} \alpha_i \bv_i
    	\right)
    \\
    &= \bb^\intercal_1
    	\cdot \left( 
    		\alpha_1 \bv_1 + 
    		\lambda^2 \bb_1
    	\right)
    = \frac{|B_1|^2}{n} + \lambda^2 |B_1|.
\end{align*}

The second of the terms can be bounded using the Cauchy-Schwarz inequality and the fact that the fraction of neighbours in the other community is $\frac{b}{d}$. 
Formally, we get
\begin{align*}
\sum_{x \in V_1} \left( \frac{|B_2(x)|}{d} \right)^2 
&\leq \left( \norm{P_{1,2} \bb_2}_2 \right)^2
\leq \left( \norm{P_{1,2}}_2\norm{\bb_2}_2\right)^2
\\
&= \left( \norm{P_{1,2}}_2 \sqrt{|B_2|}\right)^2
\leq \left( \sqrt{\norm{P_{1,2}}_1 \cdot \norm{P_{1,2}}_\infty} \cdot \sqrt{|B_2|}\right)^2
\\
&= \left( \frac{b}{d} \sqrt{|B_2|}\right)^2
= \frac{b^2}{d^2}|B_2|,
\end{align*}
where in the last inequality we combined Corollary~\ref{corollary:matrix-norm} with the two following observations:
\begin{itemize}
    \item $\Vert P_{1,2} \Vert_1 := \max_{1 \leq j \leq n} \sum_{i=1}^{n} |b_{ij}| = \frac{b}{d}$, since each node in the first community has exactly $b$ neighbors in the second community.
    \item $\Vert P_{1,2} \Vert_\infty := \max_{1 \leq i \leq n} \sum_{j=1}^{m} |b_{ij}| = \frac{b}{d}$, since each node in the second community has exactly $b$ neighbors in the first community and $P_{1,2} = P_{2,1}^\intercal$.
\end{itemize}

For the third and last term, which equals twice the product of the first two, we use the previously derived bounds and get the following quantity:
\begin{align*}
2 \sum_{x \in V_1} \frac{|B_1(x)|}{d} \cdot \frac{|B_2(x)|}{d}
&= 2 \norm{P_{1,1} b_1}_2 \cdot \norm{P_{1,2} b_2}_2
\\
&\leq 2 \frac{b}{d} \sqrt{|B_2|\left( \frac{|B_1|^2}{n} + \lambda^2 |B_1| \right)}.
\end{align*}

Before combining the three bounds, we recall that by hypothesis $G$ is such that $\frac{b}{d} \leq c_1 \cdot n^{-1/2}$ and $\lambda \leq c_2 \cdot n^{-1/4}$. 
Hence:
\begin{align*}
&\Ex{|B_1^{(t+1)|} \cond \bc{t}} 
\leq \frac{|B_1|^2}{n} + \lambda^2 |B_1| + \frac{b^2}{d^2}|B_2| + 2 \frac{b}{d} \sqrt{|B_2|\left( \frac{|B_1|^2}{n} + \lambda^2 |B_1| \right)}  + \left( \frac{|B_1|}{2} \right)
\\
&= \frac{|B_1|^2}{n} + \lambda^2 |B_1| + \frac{b^2}{d^2}|B_2| + 2 \frac{b}{d} \sqrt{|B_1| \cdot |B_2|\left( \frac{|B_1|}{n} + \lambda^2\right)}  + \left( \frac{|B_1|}{2} \right)
\\
&\leq \frac{|B_1|^2}{n} + c_2^2 \frac{|B_1|}{\sqrt{n}} + c_1^2 \frac{|B_2|}{n} + \frac{2 c_1}{\sqrt{n}} \sqrt{|B_1| \cdot |B_2|\left( \frac{|B_1|}{n} + \frac{c_2^2}{\sqrt{n}}\right)}  + \left( \frac{|B_1|}{2} \right)
\\
%< \frac{|B_1|^2}{n} + \frac{|B_1|}{\sqrt{n}} + \frac{|B_2|}{n} + \frac{2}{\sqrt{n}} \sqrt{|B_1| \cdot |B_2|}  + \left( \frac{|B_1|}{2} \right)\\
&= |B_1|\left( \frac{|B_1|}{n} + \frac{c_2^2}{\sqrt{n}} + \frac{c_1^2 |B_2|}{n |B_1|} + \frac{2c_1}{\sqrt{n}} \sqrt{\frac{|B_2|}{|B_1|}\left( \frac{|B_1|}{n} + \frac{c_2^2}{\sqrt{n}}\right)} + \frac{1}{2} \right)
\\
&< |B_1|\left( \frac{1}{2} - \frac{s_1}{2n} + \frac{c_2^2}{\sqrt{n}} + \frac{c_1^2 |B_2|}{n |B_1|} + \frac{2c_1}{\sqrt{n}} \sqrt{\frac{|B_2|}{|B_1|}\left( \frac{1}{2} -  \frac{s_1}{2n} + \frac{c_2^2}{\sqrt{n}}\right)} + \frac{1}{2} \right)
\\
&< |B_1|\left( \frac{1}{2} - \frac{s_1}{2n} + \frac{c_2^2}{\sqrt{n}} + \frac{2c_1}{\sqrt{n}} \sqrt{\frac{|B_2|}{|B_1|}\left( \frac{1}{2} -  \frac{s_1}{2n} + \frac{c_2^2}{\sqrt{n}} + \frac{c_1^2 |B_2|}{n |B_1|}\right)} + \frac{1}{2} \right)
\\
&< |B_1|\left(1 - \frac{s_1}{2n} + \frac{c_2^2}{\sqrt{n}} + \frac{2c_1}{\sqrt{n}} \sqrt{\frac{|B_2|}{|B_1|}\left( \frac{1}{2} -  \frac{s_1}{2n} + \frac{c_2^2}{\sqrt{n}} + \frac{c_1^2 |B_2|}{n |B_1|}\right)}\right).
\end{align*}

Thus, we finally get
\[
\Ex{|B_1^{(t+1)}| \cond \bc{t}} 
< 
|B_1| \left[
    1 - \frac{s_1}{2n} + \frac{c_2^2}{\sqrt{n}} 
    + \frac{2c_1}{\sqrt{n}} \sqrt{\frac{|B_2|}{|B_1|}\left( \frac{1}{2} - \frac{s_1}{2n} + \frac{c_2^2}{\sqrt{n}} + \frac{c_1^2 |B_2|}{n |B_1|}\right)}
\right]
\]
and, having a symmetric scenario in community $2$, that
\[
\Ex{|R_2^{(t+1)}| \cond \bc{t}} 
< 
|R_2| \left[
    1 + \frac{s_2}{2n} + \frac{c_2^2}{\sqrt{n}}
    + \frac{2c_1}{\sqrt{n}} \sqrt{\frac{|R_1|}{|R_2|}\left( \frac{1}{2} + \frac{s_2}{2n} + \frac{c_2^2}{\sqrt{n}} + \frac{c_1^2 |R_1|}{n |R_2|}\right)}
\right].
\]
\end{proof}

\lemmabiasincrease*
\begin{proof}
Let us start with the bias in the first community.
Note that our assumption on the bias implies that
\(
\frac{n}{4} \leq |B_1| \leq \frac{n - s_1}{2} < \frac{n}{2},
\)
and thus $\frac{|B_2|}{|B_1|} \leq 4$. 
Therefore, under these conditions the expectation of $|B_1^{(t+1)}|$ can be upper bounded as follows:
\begin{align*}
\Ex{|B_1^{(t+1)|} \cond \bc{t}}
&< |B_1|\left( 1 - \frac{s_1}{2n} + \frac{c_2^2}{\sqrt{n}} + \frac{2c_1}{\sqrt{n}} \sqrt{\frac{|B_2|}{|B_1|}\left( \frac{1}{2} -  \frac{s_1}{2n} + \frac{c_2^2}{\sqrt{n}} + \frac{c_1^2 |B_2|}{n |B_1|}\right)} \right)
\\
&< |B_1|\left( 1 - \frac{s_1}{2n} + \frac{c_2^2}{\sqrt{n}} + \frac{2c_1}{\sqrt{n}} \sqrt{4\left( \frac{1}{2} -  \frac{s_1}{2n} + \frac{c_2^2}{\sqrt{n}} + \frac{4 c_1^2}{n}\right)} \right)
\\
&\stackrel{(a)}{<} |B_1|\left( 1 - \frac{s_1}{2n} + \frac{c_2^2}{\sqrt{n}} + \frac{2c_1}{\sqrt{n}} \sqrt{4 \cdot \frac{1}{2}} \right)
\\
&< |B_1|\left( 1 - \frac{s_1}{2n} + \frac{c_2^2}{\sqrt{n}} + \frac{2c_1\sqrt{2}}{\sqrt{n}}\right)
\\
&= |B_1|\left( 1 - \frac{s_1}{2n} + \frac{2c_1\sqrt{2} + c_2^2}{\sqrt{n}}\right)
\\
&= |B_1|\left( 1 - \frac{s_1}{2n} + \frac{s_1}{4n}\right)
= |B_1|\left( 1 - \frac{s_1}{4n}\right),
\end{align*}
where in $(a)$ we used that $\frac{s_1}{2n} \geq \frac{c_2^2}{\sqrt{n}} + \frac{4 c_1^2}{n}$ by hypothesis.

Using the additive form of the Chernoff Bound and that $s_1 \leq \frac{n}{2}$, we get that
\begin{align*}
\Prob{}{|B_1|^{(t+1)} > |B_1|\left( 1 - \frac{s_1}{8n}\right) \cond \bc{t}}
&= \Prob{}{|B_1|^{(t+1)} > |B_1|\left( 1 - \frac{s_1}{4n}\right) + \frac{s_1 |B_1|}{8n} \cond \bc{t}}
\\
&\leq  \Prob{}{|B_1|^{(t+1)} > |B_1|\left( 1 - \frac{s_1}{4n}\right) + \frac{s_1}{32} \cond \bc{t}}
\\
&\leq  \Prob{}{|B_1|^{(t+1)} > \Ex{|B_1|^{(t+1)} \cond \bc{t}} + \frac{s_1}{32} \cond \bc{t}}
\\
&\leq \exp{(-{2 (s_1)^2}/{(32^2n)})}.
\end{align*}

Since it holds that $s_1 = n - 2 |B_1|$, then with probability 
$1 - \exp{(-2(s_1)^2/(32^2 n))}$ it holds that
\begin{multline*}
s_1^{(t+1)}
\geq n - 2 |B_1| \left( 1 - \frac{s_1}{8n}\right)
= n - (n - s_1) \left( 1 - \frac{s_1}{8n}\right)
\\
= n - n \left( 1 - \frac{s_1}{8n}\right) + s_1 \left( 1 - \frac{s_1}{8n}\right)
= \frac{s_1}{8} + s_1 - \frac{s_1^2}{8n}
\\
\geq \frac{s_1}{8} + s_1 - \frac{s_1}{16}
= s_1 \left( 1 + \frac{1}{16} \right).
\end{multline*}

The same reasoning can also be applied to the symmetric case of $s_2$.
\end{proof}

\lemmaclusteringbreakingsymmetry*
\begin{proof}
Let $\mathcal{I}$ be the event ``the initial configuration has the property that $s^{(0)}_1 \geq 4(2\sqrt{2}c_1 + c_2^2) \sqrt{n}$ and $-s^{(0)}_2 \geq 4(2\sqrt{2}c_1 + c_2^2) \sqrt{n}$.''
In Lemma~\ref{lemma:init} we proved that $\mathcal{I}$ happens with a probability that is at least constant. 
Then, starting from such configuration, we use Lemma~\ref{lemma:unconditional_bias_increase} in order to show that $s_1$ becomes greater than $\sqrt{n}\log n$ 
and $s_2$ becomes smaller than $-\sqrt{n}\log n$
within $\bigO(\log \log n)$ rounds, with constant probability.

We define a round $t$ to be \emph{successful} w.r.t.\ community $1$ if one of the two following conditions hold:
\begin{itemize}
    \item the process has not reached yet a configuration in which the bias is is multiplicatively increasing and is large enough, namely $s_1^{(t)} \geq s_1^{(t-1)} \left( 1 + \frac{1}{16} \right)$ and $s_1^{(t-1)} < \sqrt{n} \log n$;
    \item the bias was already large enough in a previous round, i.e., there exists a round $t' < t$ such that $s_1^{(t')} \geq \sqrt{n} \log n$.
\end{itemize}
The definition extends to community $2$ in a symmetric fashion.

Let $h =  4(2\sqrt{2}c_1 + c_2^2)$, $\alpha = 2(h/32)^2$, $\beta = (1 + 1/16)$, and let us define the events
\begin{quote}
$R_i^{(t)}$: ``The round $t$ is successful w.r.t.\ community $i$.''
\end{quote}
\begin{quote}
$\mathcal{K}_i$: ``The first $\log_{\beta} \log n$ rounds are successful w.r.t.\ community $i$.''
\end{quote}

Note that, after $T$ consecutive \emph{successful} rounds w.r.t.\ community $1$, 
the stochastic process reaches a configuration $\bar{c}$ such that 
$s_1 \geq  h\sqrt{n}(1 + 1/16)^{T}$ and then the probability that also the next round is successful is at least 
$1 - \exp{(-2h^2(1 + 1/16)^{2T}/(32^2))}$.
Conditioning to the event $\mathcal{I}$, we have that
\begin{align*}
\Prob{}{\mathcal{K}_1\cond \mathcal{I}} 
&= \Prob{}{\bigcap_{i=1}^{\log_{\beta} \log n} R_1^{(i)} \cond \mathcal{I}} 
\\
&= \prod_{i=1}^{\log_{\beta} \log n} \Prob{}{R_1^{(i)} \cond \bar{c}^{(i-1)} : \cap_{j=1}^{i-1} R_1^{(j)}, \mathcal{I}} 
\\
&\geq \prod_{i=1}^{\log_{\beta} \log n} (1 - e^{-2h^2(1+1/16)^{2i}/(32)^2})
\\
&= \prod_{i=1}^{\log_{\beta} \log n} (1 - e^{-\alpha\beta^{2i}}) 
\\
&= \exp{\left( \log \left( \prod_{i=1}^{\log_{\beta} \log n} (1 - e^{-\alpha\beta^{2i})} \right) \right)}
\\
&= \exp{\left( \sum_{i=1}^{\log_{\beta} \log n} \log(1 - e^{-\alpha\beta^{2i})}) \right)}
\\
&> \exp{\left( \sum_{i=1}^{\infty} \log(1 - e^{-\alpha\beta^{2i})}) \right)}
\\
&\stackrel{(a)}{=} \exp{\left(
    -\sum_{i=1}^\infty \left[ 
        e^{-\alpha\beta^{2i}}
        + \bigO(e^{-2\alpha\beta^{2i}})
    \right] 
\right)}
\\
&\stackrel{(b)}{>} \exp{\left(
    - \frac{1}{\alpha} \sum_{i=1}^\infty \left[ 
        \beta^{-2i}
        + \bigO(\beta^{-2i})
    \right]
\right)}
\\
&= \exp{\left( 
    -\frac{1}{\alpha} \left( \frac{1}{1-\beta^{-2}} \right) (1 + C)
\right)}
= e^{-\beta'},
\end{align*}
where in $(a)$ we expanded 
$\log(1 - x) = -x -\frac{x^2}{2} -\frac{x^3}{3} - \ldots$ 
using the Taylor series,
and in $(b)$ we used that
\[
e^{-\alpha \beta^{2i}} < \alpha^{-1} \beta^{-2i}
\implies
-\alpha \beta^{2i} < \log (\alpha^{-1} \beta^{-2i})
\implies
\alpha \beta^{2i} > \log (\alpha \beta^{2i}),
\]
which is always true for our values of $\alpha,\beta,i$ since we always have $\alpha\beta^{2i} > 0$;
the term $C$ appearing in the last bound is a constant due to the smaller order terms coming from the Taylor approximation. 

Note that the bias, starting from $\bigO(\sqrt{n})$,
reaches a value of $\bigO(\sqrt{n} \log n)$ after $\bigO(\log \log n)$ rounds. 
This implies that the bias reaches a value of at least $\sqrt{n} \log n$ within $\bigO(\log \log n)$ rounds with probability at least $e^{\beta'}$. 

In a completely symmetric fashion, the same holds for community $2$.
Now we compute the probability $\Prob{}{\mathcal{K}_1,\mathcal{K}_2 \cond \mathcal{I}}$ using the fact that, conditioning on the previous configuration, the events that a round is \emph{successful} w.r.t.\ community $1$ and community $2$ are independent.
\begin{align*}
\Prob{}{\mathcal{K}_1,\mathcal{K}_2 \cond \mathcal{I}}
&>\Prob{}{\bigcap_{i=1}^{\log_{\beta} \log n} R_1^{(i)} \wedge R_2^{(i)} \cond \mathcal{I}}
\\
&= \prod_{i=1}^{\log_{\beta} \log n} \Prob{}{R_1^{(i)} \wedge R_2{(i)} \cond  \bar{c}^{(i-1)} : \cap_{j=1}^{i-1} R_1^{(j)} \wedge R_2^{(j)}, \mathcal{I}}
\\
&= \prod_{i=1}^{\log_{\beta} \log n} \Prob{}{R_1^{(i)} \cond  \bar{c}^{(i-1)} : \cap_{j=1}^{i-1} R_1^{(j)} \wedge R_2^{(j)}, \mathcal{I}} \\
&\qquad \times \prod_{i=1}^{\log_{\beta} \log n} \Prob{}{R_2{(i)} \cond  \bar{c}^{(i-1)} : \cap_{j=1}^{i-1}  R_1^{(j)} \wedge R_2^{(j)}, \mathcal{I}}
\\
&= \Prob{}{\bigcap_{i=1}^{\log_{\beta} \log n} R_1^{(i)} \cond \mathcal{I}} \cdot \Prob{}{\bigcap_{i=1}^{\log_{\beta} \log n} R_2^{(i)} \cond \mathcal{I}}
\\
&\geq e^{-2\beta'} 
= \gamma_3.
\end{align*}
\end{proof}

\lemmaconvergence*
\begin{proof}
The proof has the following structure: We focus only on one of the biases and we show, in two different phases, that it will grow until it reaches the value $n-\log n$ within $\bigO(\log n)$ rounds, w.h.p. 
Then we apply the Union Bound and we show that this holds for both the biases, w.h.p. 

W.l.o.g.\ we assume that both the biases are positive. 
For any $i \in \{1,2\}$ we define $\tau'_i$ as the first round such that $s_i \geq \frac{n}{2}$ starting from a configuration such that $s_i \geq \sqrt{n} \log n$. 
Using Lemma~\ref{lemma:unconditional_bias_increase} and the hypotheses $s_i \geq \sqrt{n} \log n$ we get that, for each round $t$ such that $s^{(t)}_i \leq \frac{n}{2}$, it holds
\begin{align*}
&\Prob{}{s_i^{(t+1)} \geq (1 + 1/16)s_i \cond \bc{t}:s_i 
\geq \sqrt{n} \log n} 
\\
&\geq  1 - \exp{(-2(s_i)^2/(32^2 n))} 
\\
&\geq 1 - \exp^{-{\bigO(\log^2 n)}} 
> 1 - n^{-a_1}
\end{align*}
for any positive constant $a_1$. By iteratively applying the Union Bound we get that w.h.p. we have $\bigO(\log n)$ consecutive rounds of this multiplicative grow and thus it holds
\[
\Prob{}{\tau'_i > b_2 \log n} \geq 1 - n^{-a_2},
\]
for two suitable positive constants $a_2,b_2$. Now we define, for any $i \in \{1,2\}$, $\tau_i$ as the first round such that $s_i \geq n - \log n$ starting from a configuration such that $s_i \geq \frac{n}{2}$. 
Let us focus on community 1. 
Our assumption on the sign of the biases implies that there the minority color is \textit{blue}. 
This, together with the fact that $s_1 \geq \frac{n}{2}$, implies $|B_1| \leq \frac{n}{4}$.

Consider a configuration such that $2\log n \leq |B_1| \leq \frac{n}{4}$. Then it holds
\begin{equation}
\Prob{}{|B_{1}^{(t+1)|} > |B_{1}^{(t)}| (1 - \frac{1}{5}) \cond \bc{t}} < n^{-\Omega(1)}. \label{eq:minority-decrease}
\end{equation}
Indeed, using that $\frac{s_1}{2n} \geq \frac{n}{2} \geq \frac{c_2^2}{\sqrt{n}} + \frac{c_1^2 |B_2|}{n |B_1|}$ and Lemma~\ref{lemma:expectation} we get 
\begin{align*}
\Ex{|B_{1}^{(t+1)|} \cond \bc{t}} 
&< |B_1|\left(1 - \frac{s_1}{2n} + \frac{c_2^2}{\sqrt{n}} + \frac{2c_1}{\sqrt{n}} \sqrt{\frac{|B_2|}{|B_1|}\left( \frac{1}{2} - 
\frac{s_1}{2n} + \frac{c_2^2}{\sqrt{n}} + \frac{c_1^2 |B_2|}{n |B_1|}\right)}\right)
\\
&< |B_1|\left(1 - \frac{s_1}{2n} + \frac{c_2^2}{\sqrt{n}} + \frac{2c_1}{\sqrt{n}} \sqrt{\frac{|B_2|}{|B_1|}\cdot \frac{1}{2}}\right)
\\
&\leq |B_1|\left(1 - \frac{1}{4} + \frac{c_2^2}{\sqrt{n}} + \frac{2c_1}{\sqrt{n}} \sqrt{\frac{n}{2\log n}\cdot \frac{1}{2}}\right)
\\
&= |B_1|\left(1 - \frac{1}{4} + \frac{c_2^2}{\sqrt{n}} + \frac{2c_1}{\sqrt{4\log n}}\right)
\\
&= |B_1|\left(1 - \frac{1}{4} + \frac{c_2^2}{\sqrt{n}} + \frac{c_1}{\sqrt{\log n}}\right)
\leq |B_1|\left(1 - \frac{1}{5}\right),
\end{align*}
where the last inequality holds for a sufficient large $n$. Then, using the multiplicative form of the Chernoff Bound, we get that:
\begin{align*}
    \Prob{}{B_{1}^{(t+1)} \geq (1 - \frac{1}{25})|B_1|}
    &=\Prob{}{B_{1}^{(t+1)} \geq (1 + \frac{1}{5})(1 - \frac{1}{5})|B_1|}
    \\
    &\leq e^{-(1 - \frac{1}{5})|B_1|/125}
    \\
    &\leq e^{-(1 - \frac{1}{5})2\log n /125}
    = n^{-\gamma_2},
\end{align*}
for a positive constant $\gamma_2$. Let $\tau''_i$ be the first round such that $s_i \geq n - \log n$, starting from a configuration such that $s_i \geq \frac{n}{2}$. Equation~\eqref{eq:minority-decrease} implies that, by an application of the Union Bound,
\[
\Prob{}{\tau''_i > b_3 \log n} \geq 1 - n^{-a_3}
\]
for two suitable positive constants $a_3,b_3$. Thus, if we define $\tau_i$ as the first round such that $s_i \geq n - \log n$, starting from a configuration such that $s_i \geq \sqrt{n} \log n$, we get
\begin{align*}
    &\Prob{}{\tau_1 > (b_2 + b_3)\log n \,\bigcup\, \tau_2 > (b_2 + b_3)\log n} 
    \\
    &< \Prob{}{\tau'_1 > b_2\log n \,\bigcup\, \tau''_1 > b_3\log n \,\bigcup\, \tau'_2 > b_2\log n \,\bigcup\, \tau''_2 > b_3\log n} 
    \\
    &< \Prob{}{\tau'_1 > b_2\log n} + \Prob{}{\tau''_1 > b_3\log n} + \Prob{}{\tau'_2 > b_2\log n} + \Prob{}{\tau''_2 > b_3\log n}
    \\
    &< n^{-a_4},
\end{align*}
for a suitable positive constant $a_4$.
\end{proof}

\lemmametastability*
\begin{proof}
Let us define $X_u$ as an indicator random variable such that $X_u = 1$ if node $u$ will support the minority color of community $i$ at the next round and $X_u = 0$ otherwise; 
let $p_u$ be the probability of having $X_u = 1$.
We approximate $\sum_{u \in V} X_u$ with a Poisson random variable using Le Cam's Theorem (Theorem~\ref{thm:le-cam}). 
Thanks to Le Cam's Theorem, if $\sum_{u \in V_i} p_u^2 \leq \frac{1}{n^\epsilon}$, for some positive constant $\epsilon$, then any result that holds on the Poisson random variable with high probability will hold with high probability also for $\sum_{u \in V_i} X_u$. 

\begin{claim}
It holds that $\sum_{u \in C_i} p_u^2 = \bigO\left(\frac{\log^3 n}{n}\right)$.
\label{claim:psquared}
\end{claim}
\begin{proof}
Let $\sigma_i^-$ be the set of nodes supporting the minority color of community $i$ and $\sigma_i^+$ be the set of nodes supporting the majority one.
Let $z_u$ be the number of neighbors of node $u$ belonging to $\sigma_i^-$, 
and $\bar{z}_u$ be the number of neighbors of $u$ belonging to $\sigma_i^+$. 
Notice that $|\sigma_i^-| = \log n$ and $|\sigma_i^+| = n - \log n$.
Thus
\[
\sum_{u \in C_i} p_u^2 
= \sum_{\substack{u \in C_i,\\u \in \sigma_i^+}} p_u^2 + \sum_{\substack{u \in C_i,\\u \in \sigma_i^-}} p_u^2.
\]
 Let us analyze the two terms separately.
As for the first term we have that
\begin{multline*}
\sum_{\substack{u \in C_i,\\u \in \sigma_i^+}} p_u^2
\leq \sum_{\substack{u \in C_i,\\u \in \sigma_i^+}} \left(\frac{z_u + b}{a+b}\right)^4
\stackrel{(a)}{\leq} a \left(\frac{\log n + b}{a+b}\right)^2 + (n - \log n - a) \left(\frac{b}{a+b}\right)^4
\\
= \frac{a \log^4 n + 4ab \log^3 n + 6ab^2\log^2 n + 4ab^3\log n + ab^4 
+ b^4n - b^4 \log n - ab^4}{(a+b)^4}
\\
\leq \frac{a \log^4 n + 4ab \log^3 n + 6ab^2\log^2 n + 4ab^3\log n
+ b^4n - b^4 \log n}{a^4} 
\\
= \frac{\log^4 n}{a^4} + \frac{4b \log^3 n}{a^4} + \frac{6b^2 \log^2 n}{a^4} + \frac{4b^3 \log n}{a^4} + \frac{b^4 n}{a^4} - \frac{b^4 \log n}{a^4}
\\
\stackrel{(b)}{\leq} \frac{\log^4 n}{n^2} + \frac{4\log^3 n}{n^2} + \frac{6\log^2 n}{n^2} + \frac{4\log^2 n}{n^2} + \frac{n}{n^2} - \frac{\log n}{n^2}
= \bigO\left(\frac{1}{n}\right)
\end{multline*}
where in $(a)$ we used that at most $a$ nodes can have all the $\log n$ nodes belonging to $\sigma_i^-$ as neighbors, and in $(b)$ that $b \geq 1$, $a \geq b \sqrt{n} \geq \sqrt{n}$ by hypothesis of Theorem~\ref{theorem:main}, and thus $\frac{b}{a} \leq \frac{1}{\sqrt{n}}$.

As for the second term we have that
\begin{align*}
\sum_{\substack{u \in C_i,\\u \in \sigma_i^-}} p_u^2
&\leq \sum_{\substack{u \in C_i,\\u \in \sigma_i^-}} \left(\frac{\bar{z}_u + b}{a+b}\right)^2
\leq \sum_{\substack{u \in C_i,\\u \in \sigma_i^-}} \left(\frac{\log n + b}{a+b}\right)^2
\\
&= \log n \left(\frac{\log n + b}{a+b}\right)^2
\leq \frac{\log^3 n}{a^2} + \frac{2b\log^2 n}{a^2} + \frac{b^2\log n}{a^2}
\\
&\stackrel{(a)}{\leq} \frac{\log^3 n}{n} + \frac{2\log n}{n} + \frac{\log n}{n} 
= \bigO\left(\frac{\log^3 n}{n}\right),
\end{align*}
where in $(a)$ we used again that $a \geq \sqrt{n}$.

Finally, by combining the two bounds together, we get
\[
\sum_{u \in C_i} p_u^2 = \bigO\left(\frac{1}{n}\right) + \bigO\left(\frac{\log^3 n}{n}\right) = \bigO\left(\frac{\log^3 n}{n}\right).
\]
\end{proof}

We now show that a Poisson$(\lambda)$ random variable is upper bounded by $\bigO\left(\frac{\log n}{\log \log n}\right)$ w.h.p.\ as long as $\lambda$ is constant w.r.t.\ $n$.

\begin{claim}\label{claim:poisson-approx}
Let $X\sim \text{Poisson}(\lambda)$ where $\lambda$ is a positive real number, that is 
\[
\Prob{}{X=i}=\frac{\lambda^{i}}{i!}e^{-\lambda}.
\]
If $t=c\log n/\log\log n$ for some constant $c>0$ and $\lambda$ is constant w.r.t.\ $n$, then
\[
\Prob{}{X>t}\leq n^{-c+o\left(1\right)}.
\]
\end{claim}

\begin{proof}
We have
\begin{align*}
\Prob{}{X>t} 
&=\sum_{i=t+1}^{\infty}\frac{\lambda^{i}}{i!}e^{-\lambda}
=\frac{\lambda^{t}}{t!}e^{-\lambda}\sum_{i=1}^{\infty}\frac{\lambda^{i}}{\prod_{j=1}^{i}\left(t+j\right)}
\leq\frac{\lambda^{t}}{t!}e^{-\lambda}\sum_{i=1}^{\infty}\frac{\lambda^{i}}{t^{i}}
\\
&\stackrel{\left(a\right)}{\leq}\frac{\lambda^{t}}{t!}e^{-\lambda}\sum_{i=1}^{\infty}2^{-i}
\\
&\stackrel{(b)}{\leq}\left(\frac{\lambda e}{t}\right)^{t}e^{-\lambda}
\\
&=\left(\frac{\lambda e}{c\frac{\log n}{\log\log n}}\right)^{c\frac{\log n}{\log\log n}}e^{-\lambda}
\\
&=\left(e^{\log\left(\lambda e\right)-\log c-\log\log n+\log\log\log n}\right)^{c\frac{\log n}{\log\log n}}e^{-\lambda}
\\
&=e^{-\lambda+c\frac{\log n}{\log\log n}\left(\log\frac{\lambda e}{c}-\log\log n+\log\log\log n\right)}
\\
&=e^{-c\log n\left(1-o\left(1\right)\right)}
=n^{-c+o\left(1\right)},
\end{align*}
where in $(a)$ we used $t\geq2\lambda$, and in $(b)$
we used Stirling's formula $t!\geq\left(\frac{t}{e}\right)^{t}$.
\end{proof}

As a last step, we show that $\lambda = \sum_{u \in V_i} p_u$ is bounded by a constant w.r.t.\ $n$.
\begin{claim}
Let $\lambda = \sum_{u \in V_i} p_u$ and $\sigma_i^-$ be the set of nodes supporting the minority color of community $i$ and $\sigma_i^+$ be the set of nodes supporting the majority one.
It holds that $\lambda = \bigO(1)$.
\end{claim}
\begin{proof}
Let $z_u$ be the number of neighbors of node $u$ belonging to $\sigma_i^-$, 
and $\bar{z}_u$ be the number of neighbors of $u$ belonging to $\sigma_i^+$. 
Notice that $|\sigma_i^-| = \log n$ and $|\sigma_i^+| = n - \log n$.
Thus
\[
\lambda = \sum_{u \in C_i} p_u 
= \sum_{\substack{u \in C_i,\\u \in \sigma_i^+}} p_u + \sum_{\substack{u \in C_i,\\u \in \sigma_i^-}} p_u.
\]

Let us analyze the two terms separately.
As for the first term, similarly to Claim \ref{claim:psquared}, we have that
\begin{align*}
\sum_{\substack{u \in C_i,\\u \in \sigma_i^+}} p_u
&\leq \sum_{\substack{u \in C_i,\\u \in \sigma_i^+}} \left(\frac{z_u + b}{a+b}\right)^2
\\
&\stackrel{(a)}{\leq} a \left(\frac{\log n + b}{a+b}\right)^2 + (n - \log n - a) \left(\frac{b}{a+b}\right)^2 
\\
&= \frac{a \log^2 n + 2ab \log n + ab^2 + b^2n - b^2 \log n - ab^2}{(a+b)^2}
\\
&\leq \frac{a \log^2 n + 2ab \log n + b^2n - b^2 \log n}{a^2} 
\\
&= \frac{\log^2 n}{a} + \frac{2b\log n}{a} + \frac{b^2 n}{a^2} - \frac{b^2 \log n}{a^2}
\\
&\stackrel{(b)}{\leq} \frac{\log^2 n}{\sqrt{n}} + \frac{2\log n}{\sqrt{n}} + \frac{n}{n} - \frac{\log n}{n}
\\
&= 1 + \frac{\log^2 n}{\sqrt{n}} + \frac{2\log n}{\sqrt{n}} - \frac{\log n}{n},
\end{align*}
where in $(a)$ we used that at most $a$ nodes can have all the $\log n$ nodes belonging to $\sigma_i^-$ as neighbors, and in $(b)$ that $b \geq 1$, $a \geq b \sqrt{n} \geq \sqrt{n}$ by hypothesis of Theorem~\ref{theorem:main}, and thus $\frac{b}{a} \leq \frac{1}{\sqrt{n}}$.

As for the second term we have that
\begin{multline*}
\sum_{\substack{u \in C_i,\\u \in \sigma_i^-}} p_u
\leq \sum_{\substack{u \in C_i,\\u \in \sigma_i^-}} \left(\frac{\bar{z}_u + b}{a+b}\right)
\leq \sum_{\substack{u \in C_i,\\u \in \sigma_i^-}} \left(\frac{\log n + b}{a+b}\right)
= \log n \left(\frac{\log n + b}{a+b}\right)
\\
\leq \frac{\log^2 n}{a} + \frac{\log n}{a}
\stackrel{(a)}{\leq} \frac{\log^2 n}{\sqrt{n}} + \frac{\log n}{\sqrt{n}},
\end{multline*}
where in $(a)$ we used again that $a \geq \sqrt{n}$.

Finally, by combining the two bounds together, we get
\begin{multline*}
\lambda = \sum_{u \in C_i} p_u 
= \sum_{\substack{u \in C_i,\\u \in \sigma_i^+}} p_u + \sum_{\substack{u \in C_i,\\u \in \sigma_i^-}} p_u
\\
\leq 1 + \frac{\log^2 n}{\sqrt{n}} + \frac{2\log n}{\sqrt{n}} - \frac{\log n}{n} + \frac{\log^2 n}{\sqrt{n}} + \frac{\log n}{\sqrt{n}}
\\
= 1 + \frac{2\log^2 n}{\sqrt{n}} + \frac{3\log n}{\sqrt{n}} - \frac{\log n}{n}
= 1 + o(1) = \bigO(1).
\end{multline*}
\end{proof}

 We showed that the number of wrongly colored nodes $\sum_{u \in V} X_u$ is well approximated by a Poisson random variable and such random variable, thanks to Claim~\ref{claim:poisson-approx}, will be $\bigO\left(\frac{\log n}{\log \log n}\right)$ w.h.p.
\end{proof}

\section{Omitted Proofs of Section~\ref{sec:csl}}

\lemmaconsensusbreakingsymmetry*
\begin{proof}
We are interested in bounding the hitting times $\tau_1$ and $\tau_2$ defined as, respectively, the first round such that $\vert s_1 \vert \geq \sqrt{n} \log n$ and the first round such that $\vert s_2 \vert \geq \sqrt{n} \log n$.
In order to bound one of the two hitting times we use Lemma~\ref{lemma:symmetrygeneric}, a general tool for Markov chains (see~\cite[Lemma 4.5]{clementi_tight_2017}).
Let $\Omega$ be the the configuration space of the process and $m = \sqrt{n} \log n$ the target value. We need to show that the following two properties hold for each $i \in \{1,2\}$:
\begin{itemize}
    \item For any positive constant $h$, there exists a positive constant $c_1 < 1$ such that for every $x \in \Omega: s_i < m$ we have
\[
\Prob{}{s_i^{(t+1)} < h\sqrt{n} \cond X_{t} = x} < c_1,
\]
    \item There exist two positive constants $\epsilon$ and $c_2$ such that for every $x \in \Omega: h\sqrt{n} \leq s_i < m$ we have
\[
\Prob{}{s_i^{(t+1)} < (1+\epsilon)s_i\cond X_{t} = x} < e^{-c_2 s_i^2/n}.
\]
\end{itemize}

As for the first point, its proof is analogous to the proof of Lemma~\ref{lemma:init} since it is a consequence of the Berry-Esseen Theorem and of the variance of the process. 
As for the second point, we already proved it in Lemma~\ref{lemma:unconditional_bias_increase}, for $\epsilon = \frac{1}{16}$ and $c_2 = \frac{2}{32^2}
$. 
Thus we can conclude that $\Prob{}{\tau_1 > a \log n} < n^{-b}$ for  two positive constants $a,b$. 
Using the Union Bound, it is immediate to show that both the hitting times are lower bounded by $a \log n$, w.h.p.:
\begin{multline*}
    \Prob{}{\tau_1 \leq a \log n, \tau_2 \leq a \log n} = 1 - \Prob{}{\tau_1 > a \log n \,\cup\, \tau_2 > a \log n} \\
    \geq 1 - \left[ \Prob{}{\tau_1 > a \log n} + \Prob{}{ \tau_2 > a \log n} \right] 
    = 1 - 2n^b.
\end{multline*}

Note that, once a bias has reached a value of at least $\sqrt{n} \log n$, by an application of Lemma~\ref{lemma:unconditional_bias_increase} and using the hypothesis $s_i \geq \sqrt{n}\log n$ and of the Union Bound, it follows that the bias remains above that value for $\Omega(\log n)$ rounds w.h.p. This means that the system reaches a configuration such that both the biases have value at least $\sqrt{n} \log n$ within $\bigO(\log n)$ rounds w.h.p.
\end{proof}

\section{Mathematical Tools}
\begin{corollary}[H\"older]\label{corollary:matrix-norm}
Given $M \in \mathbb{R}^{n \times m}$ it holds that
\(
	\Vert M \Vert_2 \leq \sqrt{\Vert M \Vert_1 \cdot \Vert M \Vert_\infty}.
\)
\end{corollary}
\begin{proof}
We have that 
$\Vert M \Vert_2 := \sup_{\Vert \bx \Vert_2 = 1} \Vert M \bx \Vert_2$,
with $M\bx = (m_1, \ldots, m_n)^\intercal$.
Notice that
\[
	\Vert M \Vert_2^2 = \sum_{k=1}^{n} |m_k|^2
	= \sum_{k=1}^{n} \left(|m_k| \cdot |m_k|\right)
	\leq (\sup_{i} |m_i|) \cdot \sum_{k=1}^{n} |m_k|
	= \Vert M \Vert_\infty \cdot \Vert M \Vert_1.
\]
Thus, taking the square root, it follows that
$\Vert M \Vert_2 \leq \sqrt{\Vert M \Vert_1 \cdot \Vert M \Vert_\infty}$
for any vector $\bx$.
The proof is a special case of H\"older's inequality, with $p = 1$ and $q=\infty$.
\end{proof}

\begin{theorem}[Berry-Esseen]\label{thm:berry-esseen}
Let $X_1,\ldots,X_n$ be independent and identically distributed random variables with mean $\mu=0$, variance $\sigma^2 > 0$, and third absolute moment $\rho < \infty$.
Let $Y_n = \frac{1}{n}\sum_{i=1}^{n}X_i$; 
let $F_n$ be the cumulative distribution function of $\frac{Y_n\sqrt{n}}{\sigma}$;
let $\Phi$ the cumulative distribution function of the standard normal distribution.
Then, there exists a positive constant $C < 0.4748$ (see~\cite{shevtsova2014absolute} for details) such that, for all $x$ and for all $n$, 
\[
\vert F_n(x) - \Phi(x) \vert \leq \frac{C \rho}{\sigma^3 \sqrt{n}}.
\]
\end{theorem}

\begin{theorem}[Le Cam]\label{thm:le-cam}
Let $X_1,\ldots,X_n$ be independent Bernoulli random variables and let $p_i$ the probability of having $X_i = 1$.
Let $\lambda = \sum_{i=1}^{n} p_i$ and let $Y = \sum_{i=1}^{n} X_i$ be a Poisson random variable.
Then
\[
\sum_{k=0}^{\infty} \left\vert \Prob{}{Y=k} - \frac{\lambda^k e^{-\lambda}}{k!} \right\vert 
<2 \sum_{i=1}^{n} p_i^2.
\]
\end{theorem}

\begin{lemma}[Lemma 4.5~\cite{clementi_tight_2017}]
\label{lemma:symmetrygeneric}
Let $\{X_{t}\}_{t\in \mathbb{N}}$ be a Markov Chain with finite state space $\Omega$ and let 
$f:\Omega\mapsto[0,n]$ be a function that maps states to integer values. 
Let $c_3$ be any positive constant
and let $m = c_3\sqrt{n}\log n$ be a target value.
Assume the following properties hold:
\begin{enumerate}
\item For any positive constant $h$, there exists a positive constant $c_1 < 1$ such that for any $x \in \Omega : f(x) < m$,
\[
\Prob{}{f(X_{t+1}) < h\sqrt{n} \cond X_{t} = x} < c_1.
\]

\item There exist two positive constants $\epsilon, c_2$ such that for any $x \in \Omega: h\sqrt{n} \leq f(x) < m$,
\[
\Prob{}{f(X_{t+1}) < (1+\epsilon)f(X_{t})\cond X_{t} = x} < e^{-c_2f(x)^2/n}.
\]
\end{enumerate}
Then the process reaches a state $x$ such that $f(x) \geq m$ within 
$\bigO(\log n)$ rounds, w.h.p.
\end{lemma}

\begin{theorem}[Chernoff -- Additive]\label{thm:chernoff-add}
Let $X_1,\ldots,X_n$ be independent Bernoulli random variables, let $X = \sum_{i=1}^{n}X_i$, and let $\Ex{X} = \mu$. Then:
\[
\begin{array}{ll}
    \Prob{}{X \leq \mu - \lambda} \leq e^{-2\lambda^2/n},
    & 0 < \lambda < n-\mu; 
    \\
    \Prob{}{X \geq \mu + \lambda} \leq e^{-2\lambda^2/n},
    & 0 < \lambda < \mu.
\end{array}
\]
\end{theorem}

\begin{theorem}[Chernoff -- Multiplicative]\label{thm:chernoff-mult}
Let $X_1,\ldots,X_n$ be independent Bernoulli random variables, let $X = \sum_{i=1}^{n}X_i$, and let $\Ex{X} = \mu$. Then:
\[
\begin{array}{ll}
    \Prob{}{X \leq (1 - \delta) \mu} \leq e^{-\delta^2\mu/2},
    & 0 \leq \delta \leq 1; 
    \\
    \Prob{}{X \geq (1 + \delta) \mu} \leq e^{-\delta^2\mu/3},
    & 0 \leq \delta \leq 1.
\end{array}
\]
\end{theorem}

\end{document}